\documentclass[letterpaper,11pt]{article}

\usepackage{fullpage}

\usepackage{amssymb}
\usepackage{amsmath}
\usepackage{subcaption}
\usepackage[round,authoryear]{natbib}
\usepackage{authblk}
\usepackage{colortbl}
\definecolor{webgreen}{rgb}{0,0.4,0}
\definecolor{webbrown}{rgb}{0.6,0,0}
\definecolor{purple}{rgb}{0.5,0,0.25}
\definecolor{darkblue}{rgb}{0,0,0.7}
\definecolor{darkred}{rgb}{0.7,0,0}
\usepackage{hyperref}
\hypersetup{colorlinks,citecolor=darkblue,filecolor=black,linkcolor=darkred,urlcolor=webgreen,pdfpagemode=None,
pdfstartview=FitH}

\newcommand{\ignore}[1]{}

\newcommand{\estnum}{\hat{N}}
\newcommand{\estp}{\hat{p}}

\newcommand{\beat}{\text{Beat}}
\newcommand{\nbr}{\text{Nbr}}
\newcommand{\given}{\ | \ }
\newcommand{\whp}{\textit{w.h.p}}
\newcommand{\width}{0.37}
\newcommand{\surp}{{\tt surp}}
\newcommand{\mpfb}{{\tt MPFB}}

\newcommand{\plu}{{\tt Plu}}
\newcommand{\bor}{{\tt Bor}}
\newcommand{\veto}{{\tt Vet}}

\newtheorem{theorem}{{\sc Theorem}}
\newtheorem{definition}{{\sc Definition}}

\newtheorem{example}{{\sc Example}}

\newtheorem{assumption}{{\sc Assumption}}
\usepackage{cleveref}
\crefname{claim}{claim}{claims}
\crefname{fact}{fact}{facts}
\crefname{algorithm}{algorithm}{algorithms}
\crefname{observation}{observation}{observations}
\crefname{equation}{equation}{equations}
\crefname{assumption}{assumption}{assumptions}

\newenvironment{proof}{\noindent {\em Proof\/}:\enspace}
{\hfill $\blacksquare{}$ \medskip \\}

\DeclareMathOperator*{\argmax}{\arg\!\max}

\usepackage{pgf}
\usepackage{verbatim}
\usepackage{enumerate}
\usepackage{amssymb}
\usepackage{mathrsfs}
\usepackage{algorithm}
\usepackage[noend]{algorithmic}

\usepackage{resizegather}
\usepackage{url}

\usepackage{xspace}

\newcommand{\Nvec}{\vec{N}}

\newcommand{\EB}{\ensuremath{\mathbb E}\xspace}

\newcommand{\RB}{\ensuremath{\mathbb R}\xspace}

\newcommand{\OO}{\ensuremath{\mathcal O}\xspace}

\usepackage{nicefrac}

\newcommand{\nfrac}{\nicefrac}

\usepackage{enumitem}

\title{{\bf Surprise in Elections}}

\author[1]{Palash Dey}
\author[2]{Pravesh K. Kothari}
\author[3]{Swaprava Nath}

\affil[1]{\small Tata Institute of Fundamental Research, \texttt{palash.ju.dey@gmail.com}}
\affil[2]{\small Princeton University, \texttt{kothari@cs.princeton.edu}}
\affil[3]{\small Indian Institute of Technology Kanpur, \texttt{swaprava@cse.iitk.ac.in}}

\date{}
\begin{document}
\maketitle

\begin{abstract}
Elections involving a very large voter population often lead to outcomes that surprise many. This is particularly important for the elections in which results affect the economy of a sizable population. A better prediction of the true outcome helps reduce the surprise and keeps the voters prepared. 
This paper starts from the basic observation that individuals in the underlying population build estimates of the distribution of preferences of the whole population based on their local neighborhoods. The outcome of the election leads to a surprise if these local estimates contradict the outcome of the election for some fixed voting rule. To get a quantitative understanding, we propose a simple mathematical model of the setting where the individuals in the population and their connections (through geographical proximity, social networks etc.) are described by a random graph with connection probabilities that are biased based on the preferences of the individuals. Each individual also has some estimate of the bias in their connections.

We show that the election outcome leads to a surprise if the discrepancy between the estimated bias and the true bias in the local connections exceeds a certain threshold, and confirm the phenomenon that surprising outcomes are associated only with {\em closely contested elections}. We compare standard voting rules based on their performance on surprise and show that they have different behavior for different parts of the population. It also hints at an impossibility that a single voting rule will be less surprising for {\em all} parts of a population. Finally, we experiment with the UK-EU referendum (a.k.a.\ Brexit) dataset that attest some of our theoretical predictions.
\end{abstract}

\section{Introduction}
\label{sec:intro}

Recent times have witnessed quite a few elections whose outcomes are widely considered as surprises. News reports covered the unprecedented impact on trade, national economies, and job markets because of the results of the elections (e.g., Brexit~\citep{News2016}, US presidential elections~\citep{Independent2016}, UK parliamentary election~\citep{News2017,News2017a} etc.) particularly because many people and the market were unprepared for such an outcome. It has impacted not only the economy and made the stock markets unpredictable, the social impact was also paramount. It was clear that the social connections -- either online or offline -- and the mass communication media -- print or electronic -- that are important factors in opinion building, have a localized effect which does not give a holistic idea of the outcome of an election. 
This effect is more prominent in the online social media, since communities in social media inevitably group similar people together and it is easy to ignore biases. Having a large number of friends on an online social network may solidify the belief that the local observation is quite a representative sample than what actually is true. 
This raises a natural question:
% \begin{quote}
%  {\em ``Is the social network structure or the voter's perception model to predict the winner to be blamed for surprise in an election?''}
% \end{quote}
\begin{quote}
{\em ``Can the surprise/shock in an election be explained by the social network structure or the biases in the perception of the voters?''}
\end{quote}
In this paper, we address this question by proposing a model of the social network formation and voters' perception of the winner. We show that the answer cannot be obtained from an analysis that focuses on only the network structure or only the voter perception. For instance, if we consider only network structure, the following example shows that any perception about the connection probability will always leave at least half the population surprised.
\begin{example}[Limitation of a structure-based conclusion]
\label{ex:structure}
 Suppose in a population of $n$ (even) voters with two candidates (red and blue), $\nfrac{n}{2}$ are red (meaning they prefer red over blue) and the rest $\nfrac{n}{2}$ are blue. The voting rule is plurality. Suppose the network structure is such that each voter is connected with every other voter that has the same color as hers, but is connected to exactly $\nfrac{n}{2}-1$ voters of the other color. If she perceives the winner just by counting the majority at her own neighborhood, then every voter will `think' that her favorite candidate wins, and no matter how the winning candidate is chosen, half the population will always be surprised at the outcome.
 
 Clearly, the example can be adapted if the voters discount the number of voters of their own color (given the fact that they are more likely to be connected with a similar colored voter) to yield the same conclusion. Moreover, if there are more than two candidates, an extension of the construction above will lead to a surprise of the voters in the classes where the actual winner (in plurality voting over all voters) is not their favorite candidate.
\end{example}
So, it is clear that a worst case analysis over the social network structure will always lead to surprise in election -- which is hardly the case in practice -- elections with unsurprising outcomes are in fact quite normal. Later in the paper, we discuss how error in voter perception {\em alone} also cannot give rise to surprise. 
Our approach takes into account both these factors simultaneously and provides conditions when a typical voter is surprised or not. In fact, there are some counterarguments claiming that some of these elections cannot be called `surprising' given a correct model of voter perception (e.g., \cite{Economist2016} for Brexit). 

We adopt a Bayesian approach to address the question of surprise that considers the structure generation and voter perception jointly. We assume a random generative model of the voters and the social network, and show that an error in estimating the parameters of the generative process may lead to surprises.

\subsection{Our Approach and Results}
\label{sec:contributions}

Let us define the voter generation and social network formation process a bit more formally. Consider a set of $m$ candidates and $n$ voters. A class of a voter is identified by a specific linear order over the candidates -- hence there are $m!$ classes. Each voter is picked {\em i.i.d.} from a fixed probability distribution of belonging to a class. Once the voters are generated, social network among the voters are formed according to a stochastic block model. This is a general version of an Erd\"os-Renyi random graph model, where the vertices are partitioned into classes and the edge creation probabilities (which can be different) are defined only among the classes -- hence every node of a class connects to every other node in another class with the same probability. In our model, an intra-class connection probability $p_{ii}$ is assumed to be larger than an inter-class connection probability $p_{ij}$ (where $i$ and $j$ are indices for classes). For a specific voting rule $r$, e.g., plurality, and a realization of the voters denoted by the set $V$, there is a winner which we represent using $w_T(V,r)$. Since every voting rule we consider are anonymous, i.e., winner does not change even if the voter identities are changed, the winner is determined just by the number of voters in each class. Therefore, $V$ in $w_T(V,r)$ can be replaced by $N = (N_1,N_2,\ldots,N_{m!})$, where $N_j$ is the number of voters in class $j$. The perceived winner of voter $v$ is dependent on her estimates of the number of voters in different classes, denoted by $\hat{N}^v := (\hat{N}^v_1,\hat{N}^v_2,\ldots,\hat{N}^v_{m!})$, and is given by $w_P(\hat{N}^v,r)$. Voter $v$ is surprised when $w_P(\hat{N}^v,r) \neq w_T(\Nvec,r)$. We call {\em surprise} to be the probability of this event. Voter $v$ estimates $\hat{N}^v_j$ by taking the ratio of her observed neighbors of class $j$ with her estimated connection probability with class $j$. This estimation neutralizes her observation bias had the estimates been perfect. 

With this setup, our first result (\Cref{thm:2-alt-surprise}) shows that for $m=2$, if a ratio of the estimated connection probabilities stay within a threshold, a voter is {\em not} surprised with high probability (i.e., surprise asymptotically approaching zero as $n \to \infty$). However, if the threshold is crossed, the voter is surprised \whp. A corollary of this result is that if the original distribution of the voters was very biased towards one class ({\em `overwhelming majority for one candidate'}), then, even with erroneous connection probability estimates, a voter will never be surprised \whp. This result shows that voters' perception error is not solely responsible for surprise. Together with \Cref{ex:structure}, we conclude that social connection and voter perception are intertwined reasons for surprise in elections.

Having observed that {\em surprise is a phenomenon of a closely contested election}, we generalize our results for more than two candidates. As a first approach, we present the case with three candidates in \S\ref{sec:three-candidates}. However, the method clearly generalizes with similar assumptions to similar conclusions with more candidates. Unlike the case with two candidates, for three candidates, one can consider different voting rules and compare their performances w.r.t.\ surprise. We consider three prominent voting rules (that are scoring rules). Our next result (\Cref{thm:3-candidate}) shows that for different classes of voters, different rules perform better in terms of surprise -- and hints that there may not be a single surprise-optimal voting rule for all classes of voters. However, we find it interesting that the performance is not proportional to the distribution of the mass in the scoring rules since in certain class of the voters, both plurality and veto perform better than Borda voting. All voting rules are explained when presented.

Though the theoretical results in \S\ref{sec:theory} use the estimates of the connection probabilities and show that the correctness of those estimates w.r.t.\ the true values may surprise a voter, we do not explicitly mention how the voters arrive at these estimates. In \S\ref{sec:empirical}, we consider a real dataset (UK-EU referendum, a.k.a.\ Brexit) and consider a realistic model of network formation and voters' winner anticipation, that is a realistic instantiation of our theoretical model. We investigate the effect of intra and inter-class connection probabilities, and the effect of noisy observation of their estimates on surprise. We find that the conclusions in those results show a resemblance with some of the theoretical predictions.
We present the proofs in an online appendix \citep{Authors2018} due to page limitation.

\section{Related Work}
\label{sec:related}

Public elections and their outcomes had been one of the cornerstones of research in social choice theory. In the computational social choice and multi-agent systems literature, there had been several notions to measure the `goodness' of elections. For example, {\em margin of victory}, defined as the smallest number of voters who can alter the outcome of an election by voting differently~\citep{xia2012computing,dey2015estimating}, provides a quantitative threshold of surprising outcomes in terms of the voter population. A related literature exists for {\em bribery in election}~\citep[e.g.]{faliszewski2006complexity,elkind2009swap,mattei2013bribery,bredereck2016large} and complexity of manipulative attacks~\citep[e.g.]{bartholdi1989computational,conitzer2007elections,faliszewski2014complexity,parkes2012complexity}.
Surprise in election, to the best of our knowledge, has not been formally studied in this literature. There is a relevant body of literature in political economy.
\citet{ely2015suspense} formally define {\em suspense} and {\em surprise} in a dynamical model and provide a design approach to maximize either of them for a Bayesian audience. The motivation for the dynamical model comes from the examples of mystery novels, political primaries, casinos, game shows, auctions, and sports. Our definition of surprise (the outcome is contrary to a voter's belief) is closely related in spirit, and is adapted to a single-shot decision. In sports tournaments, it is important to design the schedule so that the games are highly competitive and results are unpredictable \citep{dagaev2015seeding,olson2014suspense}. In fact, {\em information design}, where a social planner aims to maximize the unpredictability of a contest has been investigated in various contexts (see, e.g., a recent survey by \citet{bergemann2017information}).
But in election outcomes stability is of prime importance~\citep{pattanaik1973stability,dummett1961stability,rubinstein1980stability}. 
The social connection model in our paper is inspired by stochastic block model. This model has a long tradition of study in the social sciences and computer science~\citep{karrer2011stochastic,holland1983stochastic,wasserman1994social}. 
Therefore, in this paper, we approach the question of surprise in election using well studied models of social connection and surprise, and introduce a voter perception model to present insightful results.

\section{Model}
\label{sec:model}

Let $[k] \triangleq \{1,\ldots,k\}$. Let $N = [n]$ be the set of \emph{voters}, and $M = \{a_1,\ldots,a_m\}$ be the set of \emph{candidates}.
Every voter has an ordinal preference over the candidates, and we assume that these preference relations are total orders, i.e., transitive, anti-symmetric, and complete. We assume $m << n$, which is representative of real elections. Since the number of preference orders can be at most $m!$, we partition the voters into disjoint {\em classes} identified by $P_k, k \in C$, with $C = [m!]$ being the indices of the classes. Voters in a given class share the same preference order. Let $\Nvec := (|P_k|, k \in C)$ denote the vector of the number of voters in each class. With a slight abuse of notation, we will refer to the preference of the voters in $P_k$ also with the same notation.

Every voter is associated with class $P_j$ with probability $\epsilon_j$ independently from other voters, where $\epsilon_j \in [0,1], \ \forall j \in C$, and $\sum_{j \in C} \epsilon_j = 1$. We assume that the $\epsilon_j$'s are unknown to the voters. The association is represented by the mapping $\sigma : N \to C$, which maps the voter identities to the class indices.
A random social network is formed with these voters by a stochastic block model which is represented by a $|C| \times |C|$ symmetric matrix $P = [p_{jk}]$, where $p_{jk}$ denotes the connection probability between the classes of voters $P_j$ and $P_k$. In this connection model, the probability of connections for every voter in a class with every voter in another class is identified by a single parameter, which may change for a different pair of classes.
% We additionally assume that the diagonal elements of $P$ are equal, i.e., the intra-class connection probabilities are same for all classes. 
The resulting graph is denoted by $G = (N,E)$, where $E$ is the edge set. 
The edge creation process is independent among each other and also is independent with the voter-to-class association process.
We assume a regularity among the connection probabilities for which we need to define a distance metric.\footnote{A valid distance metric is one that is (1) non-negative, (2) symmetric, and (3) obeys triangle inequality.}
% \begin{definition}[Kendall-Tau Distance]
 The Kendall-Tau (KT) distance between two preference orderings $P_j$ and $P_k$ is the minimum number of adjacent flip of candidates needed to reach one from the other.
% \end{definition}
Clearly, this is a valid distance metric. We call the $p_{jk}$'s {\em regular} if they are monotone decreasing with increasing KT distance between $P_j$ and $P_k$ -- which means that the voters with more dissimilar preferences are less likely to be connected. We assume that a voter knows the preferences of her immediate neighbors (on the social network) perfectly, but does know the preferences of the other voters.

A voter $v \in P_j$ estimates these connection probabilities which are denoted by $\estp_{jk}$ for all $k \in C$. We assume that the voters' estimated $\estp_{jk}$'s are also regular. At this point, we do not assume a model on how the voters reach their estimates. In \S\ref{sec:empirical}, we consider a specific model of estimates for the experiments where voters take weighted average of their own observations and a noisy version of the true global distribution. The next section deals with how the errors in these estimates can affect a voters perception of the winner. We will consider only deterministic voting rules. 

\smallskip
\noindent {\bf Voters' winner perception model:}
Voter $v$ estimates the number of voters in class $P_k$ by dividing the number of her own neighbors in that class on $G$, defined as $\nbr^k_v:= \{t: (v,t) \in E, t \in P_k\}$, with her estimated $\estp_{\sigma(v)k}$. Hence voter $v$'s estimated number of voters in class $P_k$ is,
\begin{equation}
 \label{eq:estimated_N}
 \estnum_v^k = \left\{
      \begin{array}{ll}
       \frac{1}{\estp_{\sigma(v)k}} |\nbr^k_v| & \text{ if } k \neq \sigma(v), \\
        \frac{1}{\estp_{\sigma(v)\sigma(v)}} |\nbr^{\sigma(v)}_v| + 1 & \text{ otherwise,}
      \end{array}
 \right.
%  \qquad
%  \text{for } i \in P_j.
\end{equation}
Note that if the $\estp_{\sigma(v)k}$'s were accurate, by strong law of large numbers, this estimate gives the right number of voters in each class asymptotically {\em almost surely}.

The voters now have randomly realized preferences and connections with each other. Also, every voter $v$ has an estimate of the number of voters in different classes, and therefore, under a given (anonymous) voting rule $r$, her perceived winner is denoted by $w_P({\estnum}^v,r)$, where ${\estnum}^v:= (\hat{N}^v_1,\hat{N}^v_2,\ldots,\hat{N}^v_{|C|})$.
% where ${\boldsymbol p}_j:= (p_{jk}, k \in C)$ and ${\boldsymbol \estp}_j := ({\boldsymbol p}_{jk}, k \in C)$ are respectively the vectors of true and estimated connection probabilities of that voter, and ${\boldsymbol \epsilon} := (\epsilon_k, k \in C)$ is the vector of the `excess' probabilities beyond $\frac{1}{m!}$ of a voter belonging to the classes. 
The true winner for the same realization is denoted by $w_T(\Nvec,r)$.
% 
% By this definition, both the perceived and true winners are random variables, since the association of a voter to a class and the formed social network are random. 
% However, since the number of voters is large and the random variables are drawn from a stationary distribution, we will see in the following sections that the true winner converges with high probability. 
A voter is {\em surprised} when her perceived winner is different from the true winner, defined formally as follows.
\begin{definition}[Event of Surprise]
 \label{def:surprise}
 An {\em event of surprise} of a voter $v$ for a specific realization of the voter preferences and social graph is the event where the voter's perceived winner is not the true winner, i.e., the event $S_v$ such that,
 \begin{equation}
  S_v^r := \{w_P({\estnum}^v,r) \neq w_T(\Nvec,r)\}.
 \end{equation}
\end{definition}
We will call the probability of this event as {\em surprise} of voter $v$ under voting rule $r$, denote by $\surp_v^r := P(S_v^r)$.

% Throughout this paper, we will be interested in finding the reasons that give rise to a high or low probability of surprise. 
Note that, the event of surprise is specific to a voter, but every voter in a given class has same surprise in this model, while voters in different classes may have different surprises for the same parameters. 

\smallskip
\noindent {\bf Metric to compare voting rules:}
% Consider the event of surprise closely. 
Let the event of some candidate $b \ (\neq w_T(\Nvec, r))$ beating the true winner $w_T(\Nvec, r)$ be defined as $\beat_v^r(b,w_T(\Nvec, r)) := \{b \text{ beats } w_T(\Nvec, r) \text{ in } r\}$.
The event of surprise, therefore, can be written as $S_v^r = \cup_{b \neq w_T(\Nvec, r)} \beat_v^r(b,w_T(\Nvec, r))$.
For the chosen parameters, define the {\em most probable false beating candidate} as ${b^r_v}^* \in \argmax_{b \neq w_T(\Nvec, r)} P(\beat_v^r(b,w_T(\Nvec, r)))$, with ties broken arbitrarily. Using the union bound and the fact that the probability of an union of events is always larger than that of the largest probability of the individual events, we get,
\begin{equation}
 \label{eq:union-bound}
 \begin{split}
  P(S_v^r) &= \surp_v^r \in [\ell_v^r, (m-1) \ell_v^r], \\
  & \quad \text{ where } \ell_v^r = P(\beat_v^r({b^r_v}^*,w_T(\Nvec, r))).
 \end{split}
\end{equation}
It is enough to analyze the event $\beat_v^r({b^r_v}^*,w_T(\Nvec, r))$ and consider the quantity $\mpfb_v^r := \ell_v^r$, which we will call the {\em most probable false beating (MPFB)} factor, to compare between different voting rules, since surprise can vary at most by a constant factor of this MPFB factor. In the following sections, we will see that the effect of the number of voters on this factor is in the exponent. Since the number of voters is large, the conclusions on surprise are entirely dictated by the growth or decay of the MPFB factor.

\section{Theoretical Results}
\label{sec:theory}

In this section, we first analyze the setting with two candidates to get a better insight. The set of candidates is $M = \{a_1,a_2\}$ and the classes are $P_1 = a_1 \succ a_2$ and $P_2 = a_2 \succ a_1$. WLOG, we assume that $\epsilon_1 = \frac{1}{2} + \epsilon$ and $\epsilon_2 = \frac{1}{2} - \epsilon$ with $0 < \epsilon < 1/2$. For two candidates, all standard voting rules yield the same winner as the plurality rule, and therefore, we will be considering only plurality in the case of two candidates. We first show that candidate $a_1$ emerges as winner in plurality \whp.

\begin{theorem}
\label{thm:2-alt-winner}
 When voters fall in class $P_1$ and $P_2$ w.p.\ $\frac{1}{2} + \epsilon$ and $\frac{1}{2} - \epsilon$ respectively, with $0 < \epsilon < 1/2$, $P(w_T(\Nvec, \plu) = a_2) \leqslant e^{-\sqrt{n}/2}$ for sufficiently large $n$.
\end{theorem}

\begin{proof}
 Let $X_i$ denote the number of voters in $P_i, \ i \in [2]$. Hence
\[X_i = \sum_{v \in N} \mathbb{I} \{v \in P_i\}, \ i \in [2].\]
Define, $$Z := X_2 - X_1 = \sum_{v \in N} [\mathbb{I} \{v \in P_2\} - \mathbb{I} \{v \in P_1\}] =: \sum_{v \in N} Z_v.$$
Where $Z_v := \mathbb{I} \{v \in P_2\} - \mathbb{I} \{v \in P_1\}, \ v \in N$ are i.i.d. RVs taking values $-1$ w.p.\ $\frac{1}{2}+\epsilon$ and $1$ w.p.\ $\frac{1}{2}-\epsilon$. Clearly, $\{w_T(\Nvec, \plu) = a_2\} \implies \{Z \geqslant 0\}$. We see that $\EB Z = - 2 n \epsilon$. Using Hoeffding bound, we get
\begin{align*}
 \Pr (Z - \EB Z  \geqslant t) \leqslant e^{-\frac{t^2}{2n}}.
\end{align*}
Pick $t = n^{3/4}$. Then for $n \geqslant \frac{1}{16\epsilon^4}$, $\EB Z + n^{3/4} = - 2 n \epsilon + n^{3/4} \leqslant 0$. Hence, for $n \geqslant \frac{1}{16\epsilon^4}$, we get
\begin{align*}
 \Pr(w_T(\Nvec, \plu) = a_2) \leqslant \Pr (Z \geqslant 0) \leqslant \Pr (Z \geqslant \EB Z + n^{3/4}) \leqslant e^{-\frac{\sqrt{n}}{2}}.
\end{align*}
\end{proof}
Since candidate $a_1$ turns out to be the true winner \whp., we will consider only the conditional probability that $a_2$ is the perceived winner given $a_1$ being the true winner, which will approximately be equal to surprise for large $n$. 
% For simplicity of notation, we will use $p = p_{11} = p_{22}$ for the intra-class connection probability and $q = p_{12} = p_{21}$ for the inter-class connection probability. 

\begin{theorem}[Surprise for two candidates]
\label{thm:2-alt-surprise}
When voters fall in class $P_1$ and $P_2$ w.p.\ $\frac{1}{2} + \epsilon$ and $\frac{1}{2} - \epsilon$ respectively, with $0 < \epsilon < 1/2$, we have the following.
 \begin{itemize}
  \item For voter $v$ in $P_1$,
      \begin{itemize}
       \item if $\frac{\estp_{11}}{\estp_{12}} > \frac{p_{11}}{p_{12}} \frac{1/2 + \epsilon}{1/2 - \epsilon}$, then $P(w_P({\estnum}^v, \plu) = a_2 \given w_T(\Nvec, \plu) = a_1) \geqslant 1 - 2e^{-2 \left( \frac{\estp_{11} \estp_{12}}{\estp_{11} + \estp_{12}} \right)^2 \sqrt{n}}$ for large enough $n$; hence, $\surp_v^\plu \stackrel{n \to \infty}{\to} 1$, i.e., voter $v$ is surprised \whp.
       \item if $\frac{\estp_{11}}{\estp_{12}} < \frac{p_{11}}{p_{12}} \frac{1/2 + \epsilon}{1/2 - \epsilon}$, then $P(w_P({\estnum}^v, \plu) = a_2 \given w_T(\Nvec, \plu) = a_1) \leqslant e^{-2 \left( \frac{\estp_{11} \estp_{12}}{\estp_{11} + \estp_{12}} \right)^2 \sqrt{n}}$ for large enough $n$; hence, $\surp_v^\plu \stackrel{n \to \infty}{\to} 0$, i.e., voter $v$ is not surprised \whp.
      \end{itemize}
 \item For voter $v$ in $P_2$,
      \begin{itemize}
       \item if $\frac{\estp_{22}}{\estp_{21}} < \frac{p_{22}}{p_{21}} \frac{1/2 - \epsilon}{1/2 + \epsilon}$, then $P(w_P({\estnum}^v, \plu) = a_2 \given w_T(\Nvec, \plu) = a_1) \geqslant 1 - 2e^{-2 \left( \frac{\estp_{22} \estp_{21}}{\estp_{22} + \estp_{21}} \right)^2 \sqrt{n}}$ for large enough $n$; hence, $\surp_v^\plu \stackrel{n \to \infty}{\to} 1$, i.e., voter $v$ is surprised \whp.
       \item if $\frac{\estp_{22}}{\estp_{21}} > \frac{p_{22}}{p_{21}} \frac{1/2 - \epsilon}{1/2 + \epsilon}$, then $P(w_P({\estnum}^v, \plu) = a_2 \given w_T(\Nvec, \plu) = a_1) \leqslant e^{-2 \left( \frac{\estp_{22} \estp_{21}}{\estp_{22} + \estp_{21}} \right)^2 \sqrt{n}}$ for large enough $n$; hence, $\surp_v^\plu \stackrel{n \to \infty}{\to} 0$, i.e., voter $v$ is not surprised \whp.
      \end{itemize}
 \end{itemize}
\end{theorem}

\begin{proof}
 We prove the result only for the case when $v\in P_2$, since the other case is symmetric. Define $\theta=1/2 + \epsilon$. Let the random graph formed according to the stochastic model is denoted by $G = (N, E)$. For $i\in[2]$, let $X_i$ be the set of voters denoting the neighbors of $v$ that belong to class $P_i$. Hence, $v$'s estimated number of voters in classes $P_1$ and $P_2$ are $\frac{|X_1|}{\estp_{21}}$ and $\frac{|X_2|}{\estp_{22}} + 1$ respectively. The additional one voter in the estimate of $P_2$ comes from voter $v$ counting herself. Hence
\begin{align}
 \frac{|X_1|}{\estp_{21}} &= \frac{1}{\estp_{21}} \sum_{u \in N} \mathbb{I} (\{(vu) \in E\} \cap \{u \in P_1\}), \label{eq:X-def1}\\
 \frac{|X_2|}{\estp_{22}} &= \frac{1}{\estp_{22}} \sum_{u \in N \setminus \{v\}} \mathbb{I} (\{(vu) \in E\} \cap \{u \in P_2\}). \label{eq:X-def2}
\end{align}
Taking expectations over these quantities, we get, 
\begin{align*}
 \mathbb{E} \left ( \frac{|X_1|}{\estp_{21}} \right ) &= \frac{1}{\estp_{21}} \sum_{u \in N} P (u \in P_1) \cdot P ((vu) \in E) \ | \ u \in P_1)= n \ \theta  \ \frac{p_{21}}{\estp_{21}} \qquad \text{ and,} \\
 \mathbb{E} \left ( \frac{|X_2|}{\estp_{22}} \right ) &= \frac{1}{\estp_{22}} \sum_{u \in N \setminus \{v\}} P (u \in P_2) \cdot P ((vu) \in E) \ | \ u \in P_2) = (n-1) \ (1 - \theta) \ \frac{p_{22}}{\estp_{22}}.
\end{align*}
Define a new random variable, $Z := \frac{|X_2|}{\estp_{22}} + 1 - \frac{|X_1|}{\estp_{21}}$. Its expectation is
\begin{align}
 \mathbb{E} Z &= (n-1) (1 - \theta) \frac{p_{22}}{\estp_{22}} + 1 - n \theta \frac{p_{21}}{\estp_{21}} \nonumber \\
 &= (n-1) \left ( \frac{1}{2} - \epsilon \right ) \frac{p_{22}}{\estp_{22}} + 1 - n \left ( \frac{1}{2} + \epsilon \right ) \frac{p_{21}}{\estp_{21}} \nonumber \\
 &= n \left [ \left ( \left ( \frac{1}{2} - \epsilon \right ) \frac{p_{22}}{\estp_{22}} - \left ( \frac{1}{2} + \epsilon \right ) \frac{p_{21}}{\estp_{21}} \right )  + \frac{1}{n} \left( 1 - \left ( \frac{1}{2} - \epsilon \right ) \frac{p_{22}}{\estp_{22}} \right )\right ]. \label{eq:expX_i}
\end{align}

We first consider the case when $\frac{\estp_{22}}{\estp_{21}} > \frac{p_{22}}{p_{21}} \cdot \frac{1/2 - \epsilon}{1/2 + \epsilon}$.

The first term in the bracket in \Cref{eq:expX_i} is negative since $\frac{\estp_{22}}{\estp_{21}} > \frac{p_{22}}{p_{21}} \cdot \frac{1/2 - \epsilon}{1/2 + \epsilon}$, by assumption. Let $ -\ell = \left ( \frac{1}{2} - \epsilon \right ) \frac{p_{22}}{\estp_{22}} - \left ( \frac{1}{2} + \epsilon \right ) \frac{p_{21}}{\estp_{21}}$. Hence the whole expression of \Cref{eq:expX_i} is negative for $n > \max \{ 0, \left( 1 - \left ( \frac{1}{2} - \epsilon \right ) \frac{p_{22}}{\estp_{22}} \right ) / \ell \} =: n_0 $. Hence, $\mathbb{E} Z$ is negative for sufficiently large $n$. Note from \Cref{eq:X-def1,eq:X-def2} that $Z$ can also be written as the sum over the differences of the indicator functions. We will use Hoeffding's bound since the random variables in the sum are independent. The maximum of every term in that sum of indicators that represent $Z$ can be $1/\estp_{22}$ and the minimum can be $-1/\estp_{21}$, hence the maximum difference between each of the summands is $(\estp_{22} + \estp_{21}) / \estp_{22} \estp_{21}$. We have,
\begin{align}
\label{eq:Hoeffding}
 \Pr(w_P({\estnum}^v, \plu) = a_2 \given w_T(\Nvec, \plu) = a_1) \le \Pr(Z - \mathbb{E} Z > t) &\leqslant  e^{-2\left(\frac{ \estp_{22} \estp_{21}}{\estp_{22} + \estp_{21}}\right)^2 \cdot \frac{t^2}{n}}.
\end{align}
Plugging in $t = n^{3/4}$, we get that the probability of $Z > \mathbb{E} Z + n^{3/4}$ is at most $e^{-2\left(\frac{ \estp_{22} \estp_{21}}{\estp_{22} + \estp_{21}}\right)^2 \sqrt{n}}$. Let $n_1 := \inf \{n > 0 : \left ( \left ( \frac{1}{2} - \epsilon \right ) \frac{p_{22}}{\estp_{22}} - \left ( \frac{1}{2} + \epsilon \right ) \frac{p_{21}}{\estp_{21}} \right )  + \frac{1}{n} \left( 1 - \left ( \frac{1}{2} - \epsilon \right ) \frac{p_{22}}{\estp_{22}} \right ) + \frac{1}{n^{1/4}} < 0\}$. The number $n_1$ is guaranteed to exist since $\frac{\estp_{22}}{\estp_{21}} > \frac{p_{22}}{p_{21}} \cdot \frac{1/2 - \epsilon}{1/2 + \epsilon}$, by assumption. Therefore for all $n > n_1$, $Z$ is greater than a negative quantity with probability at most $e^{-2\left(\frac{ \estp_{22} \estp_{21}}{\estp_{22} + \estp_{21}}\right)^2   \sqrt{n}}$. Since $\{ Z > 0 \} \subset \{ Z > -ve \}$, we have that $\forall n > n_1$, $\Pr(Z > 0) \leqslant e^{-2\left(\frac{ \estp_{22} \estp_{21}}{\estp_{22} + \estp_{21}}\right)^2 \sqrt{n}}$.

We now consider the case when $\frac{\estp_{22}}{\estp_{21}} < \frac{p_{22}}{p_{21}} \frac{1/2 - \epsilon}{1/2 + \epsilon}$. We leverage the calculations we did for the previous case. Because of the assumption $\frac{\estp_{22}}{\estp_{21}} < \frac{p_{22}}{p_{21}} \cdot \frac{1/2 - \epsilon}{1/2 + \epsilon}$, $\mathbb{E} Z$ is positive for large $n$ (\Cref{eq:expX_i}). Using \Cref{eq:Hoeffding}, we have,
 \[\Pr(w_P({\estnum}^v, \plu) = a_2 \given w_T(\Nvec, \plu) = a_1) \le \Pr(|Z - \mathbb{E} Z| \leqslant t) \geqslant 1 - 2 e^{-2\left(\frac{ \estp_{22} \estp_{21}}{\estp_{22} + \estp_{21}}\right)^2 \cdot \frac{t^2}{n}}.\]
 This implies that the probability of $Z \geqslant \mathbb{E} Z - t$ is at least the quantity on the RHS of the above inequality. Again, plugging in $t = n^{3/4}$ and defining $n_2 := \inf \{n > 0 : \left ( \left ( \frac{1}{2} - \epsilon \right ) \frac{p_{22}}{\estp_{22}} - \left ( \frac{1}{2} + \epsilon \right ) \frac{p_{21}}{\estp_{21}} \right )  + \frac{1}{n} \left( 1 - \left ( \frac{1}{2} - \epsilon \right ) \frac{p_{22}}{\estp_{22}} \right ) - \frac{1}{n^{1/4}} > 0\}$, which is guaranteed to exist by assumption, we get the desired conclusion for all $n>n_2$. This completes the proof.
\end{proof}

\paragraph{Corollaries.}
\Cref{thm:2-alt-surprise} captures the determining factors for surprise in plurality voting. Few conclusions are in order.
\begin{enumerate}
 \item If an agent's estimated $\estp$'s were perfect, then the agent is never surprised \whp., since then the ratios will always satisfy the `not surprised' condition of \Cref{thm:2-alt-surprise}.
%  Hence, it shows that the cause of surprise in an election under this model is due to a wrong estimate of the connection probabilities.
 \item Surprise may happen when $\epsilon$ is small, i.e., the winning margin is small. This is because, the surprise-determining thresholds for $p_{jj}/p_{jk}$s in \Cref{thm:2-alt-surprise} are very close to the actual ratios $p_{jj}/p_{jk}$s and a small error of the voter in estimating these connection parameters may lead to surprise. However, when the winning margin is large, e.g., $\epsilon$ is large enough such that $\frac{p_{22}}{p_{21}} \frac{1/2 - \epsilon}{1/2 + \epsilon} < 1$ and if the $\estp$'s are also regular, i.e., $\estp_{22}>\estp_{21}$, then no agent in $P_2$ will be surprised. This shows that elections with an overwhelming majority can hardly be surprising. Surprise is a phenomenon only of a {\em closely contested election}.
\end{enumerate}

\subsection{Three Candidates}
\label{sec:three-candidates}

We now consider the problem with three candidates. In this setting, different voting rules give rise to different winners and therefore it is possible to distinguish them w.r.t.\ the surprise metric. In this section, we will compare three voting rules, namely plurality, Borda, and veto (explained below), based on the factor $\mpfb_v^r$ (\Cref{eq:union-bound}) because of the reason explained right after the equation in \S\ref{sec:model}. 
% One can easily extend the results of this section for more than three candidates with similar conclusions. 
A collection of $m$-dimensional vectors $\vec{s}_m=\left(\alpha_1,\alpha_2,\dots,\alpha_m\right)\in\mathbb{R}^m$ 
 with $\alpha_1\geqslant\alpha_2\geqslant\dots\geqslant\alpha_m$ and $\alpha_1>\alpha_m$ for every $m\in \mathbb{N}$ defines a voting rule (called scoring rule) --- a candidate receives a score of $\alpha_i$ from a vote if it is placed at the $i$-th position in that vote, and the  score of a candidate is the sum of the scores it receives from all the votes. The winners are the candidates with the maximum score. The score vectors for the plurality, Borda, and veto voting rules are $(1,0,\ldots,0)$, $(m-1,m-2,\ldots,1,0)$, and $(1,\ldots,1,0)$ respectively. Scoring rules remain unchanged if we multiply every $\alpha_i$ by any constant $\lambda>0$ and/or add any constant $\mu$. Hence, we assume without loss of generality that, for $3$ candidates, the Borda score vector is $(\nfrac{2}{3}, \nfrac{1}{3},0)$ and the veto score vector is $(\nfrac{1}{2},\nfrac{1}{2},0)$ to ensure that $\sum_{i=1}^3 \alpha_i = 1$ for all the rules.
 
 We chose these three voting rules because (1) they are most frequently used, and (2) the distribution of scores in these rules has wide variety -- the whole score concentrated at the top alternative for plurality, (almost) equally distributed for veto, and in between these two extremes for Borda.
 
 %any score vector $\overrightarrow{s_m}$, there exists a $j$ such that $\alpha_j - \alpha_{j+1}=1$ and $\alpha_k = 0$ for all $k>j$. We call such a $\overrightarrow{s_m}$ a normalized score vector. If $\alpha_i$ is $1$ for $i\in [k]$ and $0$ otherwise, then we get the $k$-approval voting rule. For the $k$-veto voting rule, $\alpha_i$ is $0$ for $i\in [m-k]$ and $-1$ otherwise. $1$-approval is called the plurality voting rule and $1$-veto is called the veto voting rule. For the Borda voting rule, we have $\alpha_i=m-i$ for every $i\in[m]$.

% Towards that we need to define surprise of a voter quantitatively. Informally, the surprise of a voter is the probability that she predicts the winner wrongly. We formalize this in \Cref{def:surprise} below.
% 
% \begin{definition}[Surprise of a Voter]\label{def:surprise}
%  Given a voting rule $r$, we call the probability that a voter $v$'s predicted winner differs from the winner of the election the surprise of $v$ under $r$ and we denote it by $\text{surprise}_r(v)$. The probability is taken over the randomness used to generate  the preferences of each voter and the social graph.
% \end{definition}

For two candidates, we have seen that surprise occurs only in closely contested elections. Hence to compare the voting rules in this section, we consider that the voters are uniformly distributed over the $|C|$ preference classes.

\begin{assumption}[Uniform Population]\label{asm:unif_pop}
 Every voter belongs to exactly one class of preference in $\{P_k : k \in C\}$ with uniform probability.
\end{assumption}
We also assume that the voters' estimates of the connection probabilities are consistently higher than their true values as the KT distance increases between the preference class of the voter and the class of her neighbor, i.e., $p_{ij}/\estp_{ij}$'s are decreasing in ${\tt dist}_{\tt KT}(P_i,P_j)$. The motivation is to capture the fact that people often consider their local neighborhood to be representative of the global population, leading to an uniform $\estp_{ij}$'s for all $i,j \in C$. Since the true connection probabilities are regular, i.e., decreasing in ${\tt dist}_{\tt KT}(P_i,P_j)$, it gives rise to a {\em monotone estimation error}.

% Finally, we need a structural assumption about the true connection probabilities and the estimated connection probabilities of the voters. We call this {\em monotone estimation error} assumption which intuitively says that voters relatively underestimate the interconnection probabilities more than the intra-connection probabilities. We formalize this in \Cref{asm:mee} below.

\begin{assumption}[Monotone Estimation Error (MEE)]\label{asm:mee}
%  Under this model of error, the ratio of the true intra-connection probability to the estimated ones is at least the corresponding ratio for the inter-connection probability, i.e. $\nfrac{p_{11}}{\estp_{11}} \geqslant \nfrac{p_{12}}{\estp_{12}}$.
%  
 The ratio of the true connection probability to the estimated one decreases with the KT distance, i.e., $\frac{p_{k \ell}}{\estp_{k \ell}} \geqslant \frac{p_{k p}}{\estp_{k p}}$ when ${\tt dist}_{\tt KT}(P_k,P_\ell) < {\tt dist}_{\tt KT}(P_k,P_p)$ when $v \in P_k$.
\end{assumption}
% \sn{we never used this assumption anywhere, I think!}
% To keep the analysis simple, we assume a special case of regular connection probabilities and MEE. There are only two kinds of connection probabilities: intra-class, denoted by $p$ and inter-class, denoted by $q<p$ -- with the inter-class probability being same for all classes. Hence, with MEE, we have $p/\estp \geqslant q/\estq$.

In the proof of our main result in this section, we will use a quantitative version of the central limit theorem due to \cite{berry1941accuracy} and \cite{esseen1942liapounoff}. The following exposition is from \cite{tao}.

\begin{theorem}[Berry-Esseen]
\label{thm:be}
 Let $X$ be a random variable with mean $\mu$, unit variance, and finite third moment. Let $Z_n = \frac{\sum_{i=1}^n X_i}{\sqrt{n}}$, where $X_i$'s are i.i.d.\ copies of $X$. Then we have $\Pr [Z_n > \lambda]= \Pr[ G > \lambda ] + \OO( \nfrac{\EB|X|^3}{\sqrt{n}})$,
uniformly for all $\lambda\in\RB$, where $G \equiv \text{Normal}(\mu,1)$, and the implied constant in $\OO(\cdot)$ is absolute and does not depend on the distribution of $X$. 
\end{theorem}
This theorem gives a quantitative guarantee on the deviation of the cumulative distribution function of the random variable $Z_n$ from that of a normal random variable with mean same as $X$ and unit variance.

% \sn{need a more understandable statement -- much of the language here is not explained well}

With the assumptions as mentioned above, we present our main result for three candidates in the following theorem. Informally, this theorem compares plurality, Borda, and veto voting rules based on \mpfb factor. Since $\surp_v^r = \Theta(\mpfb_v^r)$ (\Cref{eq:union-bound}), we conclude that a lower MPFB factor gives a lower surprise.

\begin{theorem}\label{thm:3-candidate}
 Consider $|M|=3$, and voters are generated from an uniform population. Let $v$ be any voter.
 \begin{enumerate}[label=(\roman*)]
  \item If $v$ ranks the true winner at the first position, then $\mpfb_v^\plu \leqslant \mpfb_v^\bor \leqslant \mpfb_v^\veto$ \whp.
  \item If $v$ ranks the true winner at the second position, then $\mpfb_v^\veto \leqslant \mpfb_v^\bor \leqslant \mpfb_v^\plu$ \whp.
  \item If $v$ ranks the true winner at the last position, then $\mpfb_v^\veto \leqslant \mpfb_v^\bor$ and $\mpfb_v^\plu \leqslant \mpfb_v^\bor$ \whp.
 \end{enumerate}
\end{theorem}
% \sn{discuss the utility of this theorem right after}
\noindent
{\bf Discussion:} This result gives us a fine grained information regarding the performance on surprise of different voting rules in different voter classes. It is also clear that among these standard voting rules there is no single rule that reduces surprise for all sections of voters. But we find it interesting that the performance on surprise is not proportional to the distribution of scores in the rules, since in case (iii), Borda, that has non-extreme distribution of scores performs worse than both the other two rules having extreme score distributions.
% 
% Also, the order of performance of these three voting rules based on \mpfb depends on the position where the true winner is ranked in the voter's preference. 

% To prove \Cref{thm:3-candidate}, we consider three cases based on the position where the true winner is ranked in the voter's preference and then use \Cref{thm:be} to obtain an accurate enough estimate of \mpfb in each case. We now present our formal proof.

% \sn{more polishing needed for the proof}
\smallskip
\noindent
{\em Brief sketch of the proof:} We assume WLOG, that a specific candidate wins \whp. We consider each of the three cases in the theorem separately. For every case, we prove the claim in three stages.  First, we consider the two `false beating' events where the true winner is {\em not} the perceived winner -- for which we consider the difference in the overall scores (as we are considering only scoring rules) of the other two candidates with that of the true winner. Second, to compute the probability that these two expressions are positive (which implies that these are the false beating events), we find the mean and variance of these expressions and normalize the difference expression with the standard deviation so that the Berry-Esseen theorem (\Cref{thm:be}) can be invoked. Finally, we find the maximum of the two probabilities of false-beating events to compute the $\mpfb_v^r$ for that voting rule. The claim is proved by comparing these $\mpfb$ factors.

\begin{proof}
 Let $M = \{a_1,a_2,a_3\}$. 
 We label the classes as shown in \Cref{table:pref}. Each voter belongs to class $P_k$ w.p. $\nfrac{1}{6}$ in the uniform population model (\Cref{asm:unif_pop}). WLOG, assume that the candidate $a_2$ wins the election \whp., i.e., the overall score is highest for $a_2$ in every rule, and ties are broken in favor of $a_2$.
%  Each voter belongs to class $P_k$ w.p.\ $\epsilon_k$, where $\epsilon_k \in \left[ 0, 1 \right]$ and $\sum_{k=1}^6 \epsilon_k = 1$. 
\begin{table}[h!]
\centering
 \begin{tabular}{rl||rl}
 Class & Preferences & Class & Preferences \\ \hline
 $P_1$: & $a_1 \succ a_2 \succ a_3$ & $P_4$: & $a_2 \succ a_3 \succ a_1$ \\
 $P_2$: & $a_1 \succ a_3 \succ a_2$ & $P_5$: & $a_3 \succ a_1 \succ a_2$ \\
 $P_3$: & $a_2 \succ a_1 \succ a_3$ & $P_6$: & $a_3 \succ a_2 \succ a_1$
 \end{tabular}
 \caption{Preference classes for $3$ candidates}
 \label{table:pref}
\end{table}
 Let $(s_1,s_2,0)$ be a normalized scoring rule vector with $s_1+s_2=1$ and $s_1,s_2 \geqslant 0$. Hence, the vector is $(1,0,0)$, $(2/3,1/3,0)$, and $(1/2,1/2,0)$ respectively for $\plu$, $\bor$, and $\veto$. For a voter $v$, let $\hat{s}_v(a_1), \hat{s}_v(a_2), \hat{s}_v(a_3)$ be the random variables denoting the estimated scores for the candidates $a_1, a_2,$ and $a_3$ perceived by $v$. 
 
For every rule $r$ and voter $v$, we are interested in the differences of these estimated scores, i.e., $\hat{s}_v(a_j)-\hat{s}_v(a_2), \ j = 1,3$, since a positive value of this expression implies that a false beating event has occurred. The maximum probability of these two events is $\mpfb_v^r$.

With the voters' winner perception model, each of these estimated scores of $v$ can be written as a sum over the indicator RVs that another voter belong to a specific preference class and they are connected to $v$ (with appropriate scaling with $\estp_{kl}$ if $v \in P_k$ and the other voter is in $P_l$).
%  
%  Let us assume that the number of voters be $n+1$.  %% do we need this? we can work with n-1 other voters
 Hence, we can write the difference in the estimated scores as $\hat{s}_v(a_1)-\hat{s}_v(a_2) = \sum_{u\in N \setminus\{v\}} X_{u, a_1-a_2} + \delta_{v,a_1-a_2}$ and $\hat{s}_v(a_3)-\hat{s}_v(a_2) = \sum_{u\in N \setminus\{v\}} X_{u, a_3-a_2} + \delta_{v,a_3-a_2}$, where we clearly distinguish the contribution of voter $v$ in the differences with the variable $\delta_{v,a_j-a_2}, \ j = 1,3$. We denote the summation on the RHS in each equality with the shorthand $S_{-v,a_j-a_2} := \sum_{u\in N \setminus\{v\}} X_{u, a_j-a_2}, \ j = 1,3$.
%  Let $\delta_i(a_1-a_2)$ be the voter $v$'s contribution to the score of $a_1$ minus its contribution to the score of $a_2$ if it belongs to class $P_i$ for $i\in[6]$; for example $\delta_1(a_1-a_2)=s_1-s_2$. We similarly define $\delta_i(a_3-a_2)$ for $i\in[6]$.
% 
%  \[\hat{s}_v(a_1)-\hat{s}_v(a_2) = \sum_{u\in\VV\setminus\{v\}} X_{u, a_1-a_2},\hspace{5pt}\hat{s}_v(a_3)-\hat{s}_v(a_2) = \sum_{u\in\VV\setminus\{v\}} X_{u, a_3-a_2}\]
%  
 The expression $X_{u,a_1-a_2}$ (resp.\ $X_{u,a_3-a_2}$) is the indicator random variable denoting voter $u$'s contribution to the difference in the score of $a_1$ (resp.\ $a_3$) and $a_2$ if $u$ is connected to $v$.
 We detail out the exact expressions of $X_{u,a_j-a_2}$ when we consider the following cases.

\smallskip
 \noindent {\bf Case 1: $v\in P_1$ or $v\in P_6$ (i.e., when $v$ ranks the winner at the second position): } We only consider $v\in P_1$, since the analysis for $v\in P_6$ is symmetric. For $v\in P_1$, the expression of $X_{u,a_1-a_2}$ turns out as follows for $u\in N \setminus\{v\}$.
 \begin{gather}
   \begin{align*}
  &X_{u,a_1-a_2}= \\&(s_1 - s_2) \left ( \frac{1}{\estp_{11}} \mathbb{I} (\{(u,v) \in E\} \cap \{u \in P_1\}) - \frac{1}{\estp_{12}} \mathbb{I} (\{(u,v) \in E\} \cap \{u \in P_3\}) \right ) \\
  &\quad+ s_2 \left ( \frac{1}{\estp_{12}} \mathbb{I} (\{(u,v) \in E\} \cap \{u \in P_5\}) - \frac{1}{\estp_{12}} \mathbb{I} (\{(u,v) \in E\} \cap \{u \in P_6\}) \right ) \\ 
  &\quad+ s_1 \left ( \frac{1}{\estp_{12}} \mathbb{I} (\{(u,v) \in E\} \cap \{u \in P_2\}) - \frac{1}{\estp_{12}} \mathbb{I} (\{(u,v) \in E\} \cap \{u \in P_4\}) \right ).
 \end{align*}
 \end{gather}
Note that these are i.i.d.\ random variables for $u\in N \setminus\{v\}$, whose mean and variances are as follows.
%  Taking expectation, we get
 \[\EB[X_{u,a_1-a_2}] = (s_1 - s_2) (\nfrac{p_{11}}{6\estp_{11}} - \nfrac{p_{12}}{6\estp_{12}})\geqslant 0\]
 We get the equality due to \Cref{asm:unif_pop} and the inequality due to \Cref{asm:mee}. We also have
 \begin{gather}
  \begin{align*}
  \EB[X_{u,a_1-a_2}^2] &= (s_1-s_2)^2 \left(\nfrac{p_{11}}{6\estp_{11}^2} + \nfrac{p_{12}}{6\estp_{12}^2} \right) +(s_1^2 + s_2^2) \nfrac{p_{12}}{3\estp_{12}^2}.
 \end{align*}
 \end{gather}
 Hence
 \begin{gather}
 \begin{align*}
  \lefteqn{\text{var}(X_{u,a_1-a_2}) = \EB[X_{u,a_1-a_2}^2] -  \left(\EB[X_{u,a_1-a_2}] \right)^2} \\  
  &= (s_1-s_2)^2\left(\frac{p_{11}}{6\estp_{11}^2} + \frac{p_{12}}{6\estp_{12}^2} - \left(\frac{p_{11}}{6\estp_{11}} - \frac{p_{12}}{6\estp_{12}}\right)^2\right) + \frac{p_{12}}{3\estp_{12}^2}(s_1^2 + s_2^2).
 \end{align*}
 \end{gather}
% 
%  \[\text{var}(X_{u,a_1-a_2}) = (s_1-s_2)^2\left(\nfrac{p_{11}}{6\estp_{11}^2} + \nfrac{p_{12}}{6\estp_{12}^2} - \left(\nfrac{p_{11}}{6\estp_{11}} - \nfrac{p_{12}}{6\estp_{12}}\right)^2\right) + \nfrac{p_{12}}{3\estp_{12}^2}(s_1^2 + s_2^2)\]
%  
 For $u\in N \setminus\{v\}$, define the normalized random variable
 \[ \bar{X}_{u,a_1-a_2} = \nfrac{X_{u,a_1-a_2}}{\sqrt{\text{var}(X_{u,a_1-a_2})}}. \]
 Clearly, $\EB[\bar{X}_{u,a_1-a_2}] = \nfrac{\EB[X_{u,a_1-a_2}]}{\sqrt{\text{var}(X_{u,a_1-a_2})}}$ and $\text{var}(\bar{X}_{u,a_1-a_2})=1$. We can now apply \Cref{thm:be} for large $n$ to get
 \begin{equation}
 \label{eqn:ab1}
  \begin{split}
  &\Pr[S_{-v,a_1-a_2} + \delta_{1,a_1-a_2}>0 \given v \in P_1]\\
  &= \Pr \left[\frac{S_{-v,a_1-a_2}}{\sqrt{n \text{var}(X_{u,a_1-a_2})}} + \frac{\delta_{1,a_1-a_2}}{\sqrt{n \text{var}(X_{u,a_1-a_2})}}>0 \given v\in P_1 \right]\\
  &=\Pr \left[ G_{v,a_1-a_2} >  -\nfrac{\delta_{1,a_1-a_2}}{\sqrt{n \text{var}(X_{u,a_1-a_2})}} \right] + \OO \left( \nfrac{1}{\sqrt{n}} \right)\\
  &=\Pr [ G_{v,a_1-a_2} \geqslant  0] + \OO \left( \nfrac{1}{\sqrt{n}} \right). %\numberthis 
 \end{split}
 \end{equation}
 Where $G_{v,a_1-a_2}$ is a normal RV with mean $\EB[\bar{X}_{u,a_1-a_2}]$ and unit variance. 
 The last equality follows from the fact that $\Pr \left[0 > G_{v,a_1-a_2} > -\nfrac{\delta_{1,a_1-a_2}}{\sqrt{n \text{var}(X_{u,a_1-a_2})}} \right] = \OO( \nfrac{1}{\sqrt{n}})$ since the length of the interval $\left[ -\nfrac{\delta_{1,a_1-a_2}}{\sqrt{n \text{var}(X_{u,a_1-a_2})}}, 0 \right]$ is $\OO( \nfrac{1}{\sqrt{n}})$, hence the integral of any probability distribution over it is $\OO( \nfrac{1}{\sqrt{n}})$.
 
 Similarly, for $u\in N \setminus\{v\}$, $X_{u,a_3-a_2}$ is defined as follows.
 \begin{gather}
  \begin{align*}
  &X_{u,a_3-a_2} = \\
  &(s_1 - s_2) \left ( \frac{1}{\estp_{12}} \mathbb{I} (\{(u,v) \in E\} \cap \{u \in P_6\}) - \frac{1}{\estp_{12}} \mathbb{I} (\{(u,v) \in E\} \cap \{u \in P_4\}) \right) \\
  &\quad+ s_2 \left ( \frac{1}{\estp_{12}} \mathbb{I} (\{(u,v) \in E\} \cap \{u \in P_2\}) - \frac{1}{\estp_{11}} \mathbb{I} (\{(u,v) \in E\} \cap \{u \in P_1\}) \right ) \\ 
  &\quad+ s_1 \left ( \frac{1}{\estp_{12}} \mathbb{I} (\{(u,v) \in E\} \cap \{u \in P_5\}) - \frac{1}{\estp_{12}} \mathbb{I} (\{(u,v) \in E\} \cap \{u \in P_3\}) \right ).
 \end{align*}
 \end{gather}
 Taking expectation, we get
 \[\EB[X_{u,a_3-a_2}] = s_2(\nfrac{p_{12}}{6\estp_{12}} - \nfrac{p_{11}}{6\estp_{11}})\leqslant 0\]
 The equality follows due to \Cref{asm:unif_pop} and the inequality due to \Cref{asm:mee}. 
% Since $\EB[X_{u,a_3-a_2}]$ is at most $0$, performing similar calculation as we did for $X_{u,a_1-a_2}$ above, we conclude the following.
 Performing similar calculation as we did for $X_{u,a_1-a_2}$, we reach a unit variance normal RV $G_{v,a_3-a_2}$. However, the mean of $G_{v,a_1-a_2}$ turns out to be larger than $G_{v,a_3-a_2}$, which lead to the conclusion that for large $n$
 \begin{align*}
  &\Pr[S_{-v,a_1-a_2} + \delta_{1,a_1-a_2}>0 \given v \in P_1]\\
  &\geqslant \Pr[S_{-v,a_3-a_2} + \delta_{1,a_3-a_2}>0 \given v\in P_1]
 \end{align*}
Hence, to find the MPFB factor in this case, we need to compare the probability of \Cref{eqn:ab1} among different voting rules. Since, the probability reduces to the tail distribution of $G_{v,a_1-a_2}$ which is a normal RV with unit variance, it is enough to compare the means of $G_{v,a_1-a_2}$ to compare the MPFB factors. Denoting the means of $G_{v,a_1-a_2}$ by $\mu_v^r$ for voter $v$ under rule $r$, we get for large enough $n$
 \[\mu_v^\veto \leqslant \mu_v^\bor \leqslant \mu_v^\plu. \]
 Which implies \whp.
 \[\mpfb_v^\veto \leqslant \mpfb_v^\bor \leqslant \mpfb_v^\plu.\]
%  \[\Pr[\hat{s}_v(a_1)-\hat{s}_v(a_2) + \delta_1(a_1-a_2)>0|v\in P_1] \geqslant \Pr[\hat{s}_v(a_3)-\hat{s}_v(a_2) + \delta_1(a_3-a_2)>0|v\in P_1]\]
 
 The analysis for $v\in P_6$ is the same with the roles of candidates $a_1$ and $a_3$ being reversed. Hence, we have proved claim (ii) of the theorem.
 
%  \sn{since the proofs for different cases are repetitive, I think we should consider making the case 1 as detailed as possible, and leave with the proof sketches for the other two cases -- actually we should push the case 2 and 3 to appendix}
 
%  \smallskip
%  Case 2 considers $v\in P_2$ or $v\in P_5$ (i.e., when $v$ ranks the winner at the last position -- claim (iii)) and Case 3 considers $v\in P_3$ or $v\in P_4$ (i.e., when $v$ ranks the winner at the first position -- claim (i)), and both uses similar arguments as Case 1. Due to space limitation, we skip their proofs.
% %  we postpone their proofs to the appendix.
% 
% \if 0
 \smallskip
 \noindent {\bf Case 2: $v\in P_2$ or $v\in P_5$ (i.e., when $v$ ranks the winner at the last position): } 
 We adopt a similar calculation as Case 1 to get
%  Using similar calculation as for the case when $v\in P_1$, we have the following. 
 \begin{align*}
 \EB[X_{u,a_1-a_2}] &= s_1 (\nfrac{p_{11}}{6\estp_{11}} - \nfrac{p_{12}}{6\estp_{12}})\geqslant 0\\
  \EB[X_{u,a_1-a_2}^2] &= s_1^2(\nfrac{p_{11}}{6\estp_{11}^2} + \nfrac{p_{12}}{6\estp_{12}^2}) + \nfrac{p_{12}}{3\estp_{12}^2}((s_1-s_2)^2 + s_2^2)
 \end{align*}
 With notations similar to Case 1, we denote the mean of the normalized variance normal RV $G_{v,a_j-a_2}$ by $\mu_{v,a_j-a_2}^r, \ j = 1,3$, for the differences of estimated scores of voter $v$ between candidates $a_j$ and $a_2, \ j = 1,3$.
 Hence
 \begin{align*}
  \mu_{v,a_j-a_2}^r = \EB \left [\frac{S_{-v,a_j-a_2}}{\sqrt{n \text{var}(X_{u,a_j-a_2}})} \given v\in P_2 \right], \ j =1,3.
 \end{align*}
With similar computations, we get for voter $v \in P_2$
%  Continuing in similar fashion, by computing $\mu(a_1-a_2)=\EB[\nfrac{(\hat{s}_v(a_1)-\hat{s}_v(a_2))}{\sqrt{\text{var}(X_{u,a_1-a_2}}}|v\in P_2]$ and $\mu(a_1-a_2)=\EB[\nfrac{(\hat{s}_v(a_3)-\hat{s}_v(a_2))}{\sqrt{\text{var}(X_{u,a_3-a_2}}}|v\in P_2]$ we have the following.
\begin{gather}
  \begin{align*}
  &\max\{ \mu_{v,a_1-a_2}^\plu, \mu_{v,a_3-a_2}^\plu\} =\max\{\mu_{v,a_1-a_2}^\veto,\mu_{v,a_3-a_2}^\veto\} \leqslant \max\{\mu_{v,a_1-a_2}^\bor,\mu_{v,a_3-a_2}^\bor\}
 \end{align*}
\end{gather}
 Which implies \whp.
 \begin{align*}
  &\mpfb_v^\veto \leqslant \mpfb_v^\bor \quad \text{ and } \quad \mpfb_v^\plu \leqslant \mpfb_v^\bor.
 \end{align*}
 The case for $v \in P_5$ is same with the roles of candidates $a_1$ and $a_3$ reversed. Hence, we have proved claim (iii) of the theorem.

\smallskip
 \noindent {\bf Case 3: $v\in P_3$ or $v\in P_4$ (i.e., when $v$ ranks the winner at the first position): }
%  We adopt a similar calculation as Case 1 to get
% %  Using similar calculation as for the case when $v\in P_1$, we have the following. 
%  \begin{align*}
%  \EB[X_{u,a_1-a_2}] &= s_1 (\nfrac{p_{11}}{6\estp_{11}} - \nfrac{p_{12}}{6\estp_{12}})\geqslant 0\\
%   \EB[X_{u,a_1-a_2}^2] &= s_1^2(\nfrac{p_{11}}{6\estp_{11}^2} + \nfrac{p_{12}}{6\estp_{12}^2}) + \nfrac{p_{12}}{3\estp_{12}^2}((s_1-s_2)^2 + s_2^2)
%  \end{align*}
 With notations similar to Case 1, and denoting the mean of the normalized variance normal RV $G_{v,a_j-a_2}$ by $\mu_{v,a_j-a_2}^r, \ j = 1,3$, for the differences of estimated scores of voter $v$ between candidates $a_j$ and $a_2, \ j = 1,3$, we have
 \begin{align*}
  \mu_{v,a_j-a_2}^r = \EB \left [\frac{S_{-v,a_j-a_2}}{n \sqrt{\text{var}(X_{u,a_j-a_2}})} \given v\in P_3 \right], \ j =1,3.
 \end{align*}
With similar computations, we get for voter $v \in P_3$
%  Continuing in similar fashion, by computing $\mu(a_1-a_2)=\EB[\nfrac{(\hat{s}_v(a_1)-\hat{s}_v(a_2))}{\sqrt{\text{var}(X_{u,a_1-a_2}}}|v\in P_2]$ and $\mu(a_1-a_2)=\EB[\nfrac{(\hat{s}_v(a_3)-\hat{s}_v(a_2))}{\sqrt{\text{var}(X_{u,a_3-a_2}}}|v\in P_2]$ we have the following.
 \begin{gather}
 \begin{align*}
  &\max\{ \mu_{v,a_1-a_2}^\plu, \mu_{v,a_3-a_2}^\plu\} \leqslant \max\{\mu_{v,a_1-a_2}^\veto,\mu_{v,a_3-a_2}^\veto\} \leqslant \max\{\mu_{v,a_1-a_2}^\bor,\mu_{v,a_3-a_2}^\bor\}
 \end{align*}
 \end{gather}
 Which implies \whp.
 \begin{align*}
 \mpfb_v^\plu \leqslant \mpfb_v^\bor \leqslant \mpfb_v^\veto.
 \end{align*}
 The case for $v \in P_4$ is same with the roles of candidates $a_1$ and $a_3$ reversed. Hence, we have proved claim (i) of the theorem.
%  
% \fi
\end{proof}

\section{Empirical Results}
\label{sec:empirical}

Our theoretical results in \S\ref{sec:theory} use some simplifying assumptions in the interest of a cleaner analysis. Firstly, we assumed that the voters have an estimate of the connection probabilities, though we do not explicitly mention how the voters arrive at these estimates. In practice, voters anticipate a winner by implicitly estimating the number of voters voting in favor of the candidate versus voting against him. There are typically two major sources of information to a voter: first, via her own neighbors in the (online/offline) social network, and second via the public broadcasting media -- print or electronic. Secondly, our connection model was following the stochastic block model that only depends on the voters' preferences and had no dependence on the voters' geographical locations. In this section, we relax these two simplifying assumptions and from an empirical viewpoint try to see if the broad theoretical predictions hold. 

We instantiate the voting population with a real election dataset. We construct the social network of voters depending on their preferences and geographical locations to make the network more realistic. We capture a voter $v$'s estimates of the population of different classes by taking a weighted average of (1) voter $v$'s individual observation, i.e., the number of voters of different classes in $v$'s immediate neighborhood and (2) a noisy version of the global (true) number of voters in each class. Effects (1) and (2) capture a voter's private and public observations respectively and give a realistic view of opinion forming.

% % 
% \begin{figure}[t!]
%   \centering
% %      \includegraphics[width=\textwidth]{shock_dist_4_2.eps}
%   \scalebox{0.55}{\input{shock_4_2_full.pgf}}
% \caption{Effect of weight on global observation and observation bias on surprise.}
% \label{fig:detailed}
% % \vspace{-5mm}
% \end{figure}
% % 
% % 

\begin{figure*}[h!]
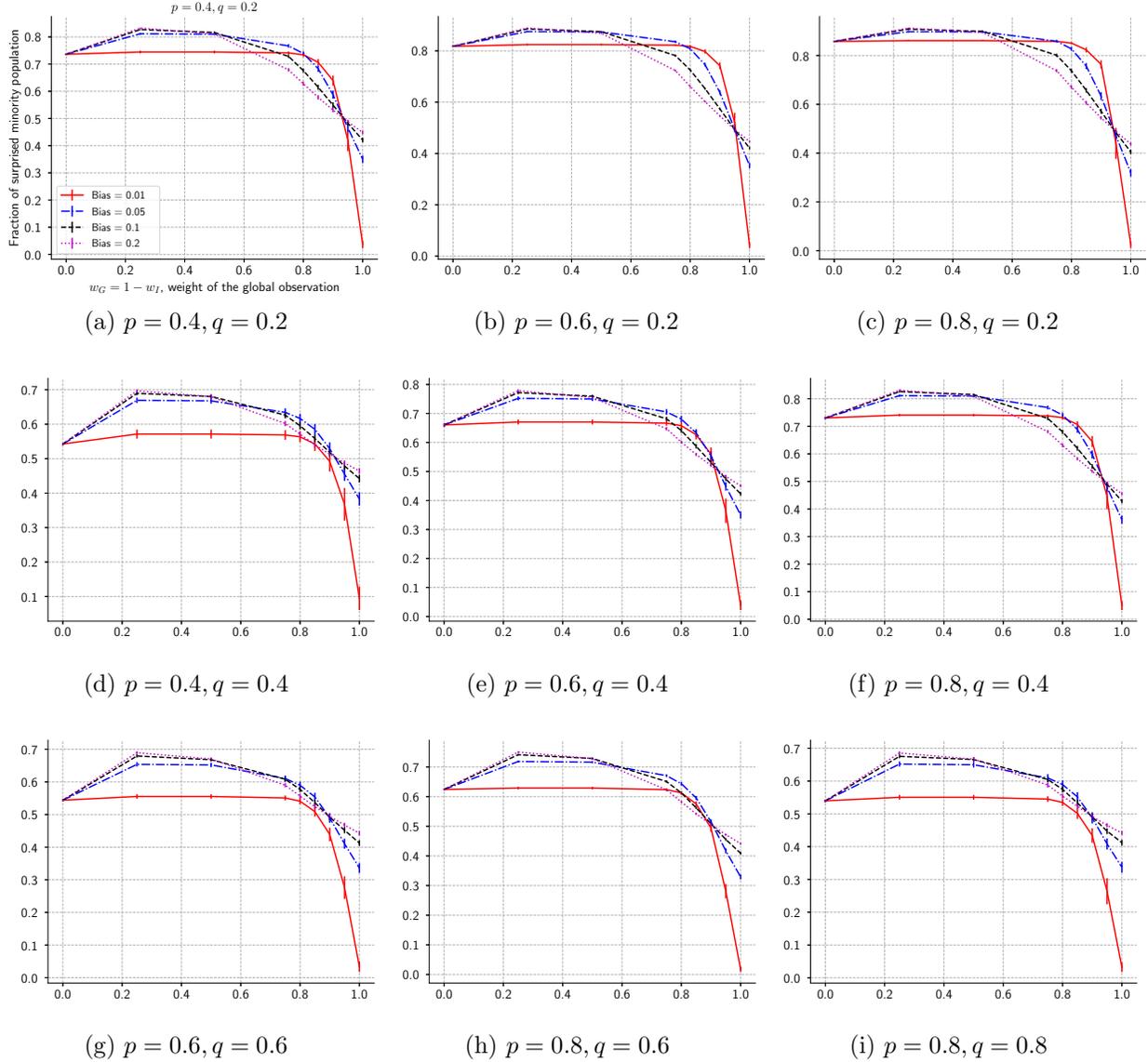

%% LINE 1
        \begin{subfigure}[b]{0.33\textwidth}
                \scalebox{\width}{
                \input{shock_4_2_full.pgf}}
                \caption{$p=0.4, q=0.2$}
                \label{subfig_4_2}
        \end{subfigure} %
        \begin{subfigure}[b]{0.33\textwidth}
                \scalebox{\width}{\input{shock_6_2_noedge.pgf}}
                \caption{$p=0.6, q=0.2$}
                \label{subfig_6_2}
        \end{subfigure}%
        \begin{subfigure}[b]{0.33\textwidth}
                \scalebox{\width}{\input{shock_8_2_noedge.pgf}}
		\caption{$p=0.8, q=0.2$}
                \label{subfig_8_2}
        \end{subfigure}%
        \\
%% LINE 2
        \begin{subfigure}[b]{0.33\textwidth}
                \scalebox{\width}{\input{shock_4_4_noedge.pgf}}
                \caption{$p=0.4, q=0.4$}
                \label{subfig_4_4}
        \end{subfigure}%
        \begin{subfigure}[b]{0.33\textwidth}
                \scalebox{\width}{\input{shock_6_4_noedge.pgf}}
                \caption{$p=0.6, q=0.4$}
                \label{subfig_6_4}
        \end{subfigure}%
        \begin{subfigure}[b]{0.33\textwidth}
                \scalebox{\width}{\input{shock_8_4_noedge.pgf}}
		\caption{$p=0.8, q=0.4$}
                \label{subfig_8_4}
        \end{subfigure}%
        \\
% LINE 3
        \begin{subfigure}[b]{0.33\textwidth}
                \scalebox{\width}{\input{shock_6_6_noedge.pgf}}
                \caption{$p=0.6, q=0.6$}
                \label{subfig_6_6}
        \end{subfigure}%
        \begin{subfigure}[b]{0.33\textwidth}
                \scalebox{\width}{\input{shock_8_6_noedge.pgf}}
                \caption{$p=0.8, q=0.6$}
                \label{subfig_8_6}
        \end{subfigure}%
        \begin{subfigure}[b]{0.33\textwidth}
                \scalebox{\width}{\input{shock_8_8_noedge.pgf}}
		\caption{$p=0.8, q=0.8$}
                \label{subfig_8_8}
        \end{subfigure}%
        \caption{Surprise in Brexit for different intra and inter-class connection probabilities (legends same as \Cref{subfig_4_2}).}
    \label{fig:shock-consolidated}
    \vspace{-4mm}
\end{figure*}

\smallskip
\noindent {\em Datasets:} We use the UK election dataset of EU referendum (popularly known as Brexit)\footnote{\tt https://goo.gl/MtTdIT}. The dataset is publicly available and gives the total count of votes cast by the UK voters that voted either `remain' (R) or `leave' (L) the EU. The data consists of approximately $33$ million valid votes and is partitioned across 382 regions within the UK. Each region is identified with the name of the town, city, or county. We will refer to this dataset as Brexit dataset. We have used another dataset\footnote{https://www.townslist.co.uk/} to find the latitude and longitude of these regions. Since the location dataset gives the latitude-longitude of a town and the voting constituencies are collection of a number of them, we have averaged over the towns in a region to find the approximate centroid of the region. There were few locations (about 18\%) whose information were not available in the location dataset, we have filled in their location to be the centroid of all the available locations in the dataset.
The Brexit data are suitable for our experiment, since (a) it has only two candidates for which we have a simple yet insightful theoretical result (\Cref{thm:2-alt-surprise}), (b) it is large enough to draw conclusions on large-scale elections, and (c) the election was closely contested, (51.9\% for L and 48.1\% for R).

\smallskip
\noindent {\em Approach:} In each location, based on the total number of voters and their votes, we re-created the voters. The connection follows a random graph model where the probability of connection between two voters is the average of (a) $p_1$, which is decreasing in the geographical distance between the voters, and (b) $p_2$, which is $p$ if both voters are from the same class, and $q$, otherwise (with $p \geqslant q$). This relaxation from the theoretical model allows for a social connection where two individuals are geographically close despite having different political opinions.

In this social network, voters perceive the outcome of the election according to effects (1) and (2) as explained before. For the individual observation (effect 1), we assume that a voter can perfectly observe the true voting preferences of her immediate neighbors in the graph. The number of voters that voted R or L in the immediate neighborhood gives a distribution of the R and L voters in the neighborhood including herself. For the global observation (effect 2), we add a zero mean truncated Gaussian noise to the {\em true} distribution of the votes -- the truncated Gaussian is set such that after the addition of noise, the resulting noisy distribution still remains a valid one, i.e., no probability mass goes negative. We call the variance of the truncated Gaussian the {\em bias} of this observation. The voter combines these two distributions with weights $w_I$ for the noise-free individual distribution and $w_G$ for the noisy global distribution. Her perceived winner is the one that has larger mass among the two outcomes in the weighted sum distribution.

Due to the massive scale of the dataset, which takes significant time to run a single experiment, we have sampled $10,000$ votes uniformly at random and created a sub-election. In this sub-election, every individual attempts to connect to $500$ other individuals picked uniformly at random. In this discussion, we consider the surprise of the voters in the minority class of this sub-election (i.e., the voters whose favorite candidate does not win -- hence they get surprised when they perceive this candidate to be the winner).

\smallskip
\noindent {\em Results:} We have three independent parameters that give rise to surprise: (1) the weight on global observation $w_G$ ($w_I$ is fixed given this), (2) the bias on this observation, and (3) the choices of $p$ and $q$. To show how these parameters affect surprise, we plot the fraction of surprised minority population versus $w_G$ for different choices of observation bias of the global distribution. \Cref{subfig_4_2} shows such a plot for a specific choice of $p$ and $q$. \Cref{subfig_4_2,subfig_6_2,subfig_8_2} show similar information when $\nfrac{p}{q}$ increases.

\smallskip
\noindent{\bf Observations.}
Some results support our theoretical predictions, even after relaxing our assumptions on network formation and voter estimates. (i) When the ratio $\nfrac{p}{q}$ is large, the surprises are large too. A large $\nfrac{p_{22}}{p_{21}}$ implies that more $\estp_{22},\estp_{21}$ satisfies the condition of surprise in part 2 of \Cref{thm:2-alt-surprise} (here $p_{22} = p, p_{21} = q$) -- giving rise to a higher surprise. (ii) More bias in the observation leads to a higher surprise (note, e.g., the $w_G = 1.0$ point in \Cref{subfig_4_2}). This too is expected by \Cref{thm:2-alt-surprise} as a larger bias gives rise to a larger chance that the estimated ratio of $p$ and $q$ will be different from $\nfrac{p}{q}$ -- thereby making the condition of surprise in part 2 of \Cref{thm:2-alt-surprise} gets satisfied more likely.
% difficult to explain

However, there are a few observations that we find surprising. The downward trend of the curve was expected with more weight on global information -- but when there is noise in the global information, there is an increase and dip in the surprise. Also, each curve shows a cross-over region, where mixing a more noisy global observation gives a lower surprise.

\section{Discussion}

Our results give a quantitative understanding of surprise in elections. We set up a model for voters' preference generation, social network creation, and voters' perception of the winner from their local neighborhood. 
% To our knowledge, this is a first attempt to understand surprise in election from a voter's perspective. 
Our results for more than two candidates hint that possibly no single voting rule can reduce the surprise for all sections of the voters.
The empirical results complement our assumption on voter's estimates of connection probabilities. However, a more fine-grained model of voter perception will help better understand the surprise phenomenon. We believe that a thorough understanding of surprise is essential for mitigating it -- particularly when such surprises affect the social, economic, and political decisions of individuals.

\subsection*{Acknowledgments}

We are grateful to Debasis Mishra and two anonymous referees for useful discussions and suggestions during an earlier version of the paper. Swaprava Nath is supported by IIT Kanpur faculty initiation grant.

% \newpage 
% Bibliography
\bibliographystyle{plainnat}
\bibliography{abb,ultimate,palash,swaprava}

\end{document}